\documentclass[a4paper,UKenglish]{lipics-v2018}

\usepackage{mathpartir}
\usepackage{amsmath}
\usepackage{amssymb}
\usepackage{amsthm}
\usepackage{todonotes}
\usepackage{trsym}
\usepackage{color}
\usepackage{verbatim}
\usepackage{url}
\usepackage{wrapfig}
\usepackage{rotating}
\usepackage{scalerel}
\makeatletter
\def\url@leostyle{%
  \@ifundefined{selectfont}{\def\UrlFont{\sf}}{\def\UrlFont{\small\ttfamily}}}
\makeatother
\usepackage[curve]{xypic}

\usepackage{graphics}
\usepackage{graphicx}
\usepackage[all]{xy}
\usepackage{eucal}
\usepackage{enumitem}

\DeclareMathOperator{\supp}{supp}

\newcommand{\id}{\mbox{\sl id}}
\newcommand{\after}{\mathrel{\circ}}

\newcommand{\Pow}{\mathcal{P}}
\newcommand{\Powne}{\mathcal{P}_{ne}}

\newcommand{\Cat}[1]{\ensuremath{\mathbf{#1}}}
\newcommand{\Sets}{\Cat{Sets}}

\newcommand{\Kl}{\mathcal{K}{\kern-.2ex}\ell}

\newcommand{\set}[2]{\{#1\;|\;#2\}}

\newcommand{\Alg}[1]{\mathbb{#1}}			
\DeclareMathOperator{\conv}{conv}			    
\DeclareMathOperator{\Ext}{Ext}			    
\DeclareMathOperator{\AlgCat}{Alg}			    
	
\DeclareMathOperator{\EM}{E{\kern-.2ex}M}

\newcommand{\Dis}{\mathcal{D}}				

\newcommand{\T}{\mathcal M}				

\newcommand{\Idmap}{\textrm{Id}}

\newcommand{\nePow}{\mathcal{P}_{ne}}

\newcommand{\cset}{\mathcal C}
\newcommand{\dset}{\mathcal D}
\newcommand{\done}{d}
\newcommand{\dtwo}{e}

\newcommand{\nplus}{\oplus}
\newcommand{\mplus}{+}

\theoremstyle{plain}
\newtheorem{proposition}[theorem]{Proposition}
\theoremstyle{remark}

\title{Presenting convex sets of probability distributions by convex semilattices and unique bases}

\author{Filippo Bonchi}{University of Pisa, Italy}{filippo.bonchi@unipi.it}{}{}
\author{Ana Sokolova}{University of Salzburg, Austria}{ana.sokolova@cs.uni-salzburg.at}{}{}
\author{Valeria Vignudelli}{CNRS/ENS Lyon, France}{valeria.vignudelli@ens-lyon.fr}{}{}
\authorrunning{Bonchi, Sokolova, Vignudelli}

\nolinenumbers
\begin{document}
\maketitle

\begin{abstract}
We prove that every finitely generated convex set of finitely supported probability distributions has a unique base, and use this result to show that the monad of convex sets of probability distributions is presented by the algebraic theory of convex semilattices. 
\end{abstract}

\section{Introduction}
Models of computations exhibiting both nondeterministic and probabilistic behaviour are abundantly used in computed assisted verification~\cite{Baier2008,HermannsKK14,KNP02,DehnertJK017,vardi1985automatic,hansson1991time,segala1995probabilistic}, Artificial Intelligence~\cite{CPP09,Kaelbling:1998vs,RussellNorvig:2009}, and studied from semantics perspective~\cite{HeunenKSY17,StatonYWHK16,HermannsPSWZ11}. 
Indeed, probability is needed to quantitatively model uncertainty and belief, whereas nondeterminism enables modelling of incomplete information, unknown environment, implementation freedom, or concurrency. 

Since several decades, computer scientists have found it convenient to exploit algebraic methods to analyse computing systems. From an algebraic perspective, the interplay of nondeterminism and probability has been posing some remarkable challenges~\cite{VaraccaW06,DBLP:journals/corr/KeimelP16,Mio14,Jacobs08,Varacca03,Mislove00,Goubault-Larrecq08a,TixKP09a}. Nevertheless, several fundamental algebraic structures have been identified and studied in depth.

In this paper we focus on one of such structures, namely \emph{convex sets of probability distributions}. These sets give rise to a monad that is well known in the literature and has found applications in several works~\cite{Mislove00,Goubault-Larrecq08a,TixKP09a, Varacca03,VaraccaW06,Jacobs08}. In recent work \cite{BSV19arxiv}, we proved that this monad is presented by the algebraic theory of \emph{convex semilattices}. In this paper, we provide an alternative proof based on a simple property: We show that every (finitely generated) convex set of distribution has a \emph{unique base}.

{\bf Synopsis:} In Section \ref{sec:ubthm}, we show the unique base theorem in its simplest formulation. We introduce the basic categorical machinery in Section \ref{sec:monads}, while in Section \ref{sec:monadC} we recall the monad of interest as well as the theory of convex semilattices. Sections \ref{secproofOne} and \ref{sec:inverse} provide our alternative proof of the presentation of the monad.

\section{A unique base theorem for convex sets of probability distributions}\label{sec:ubthm}

Given a set $X$, a probability distribution is a function $d \colon X \to [0,1]$ such that $\sum_{x\in X} d(x)=1$. A probability distribution $d$ is finitely supported if $d(x)\neq 0$ for finitely many $x$. We call $\mathcal{D}(X)$ the set of finitely suported probability distributions over $X$.

A probability distribution $d\in \mathcal{D}(X)$ is a convex combination of the distributions $d_1,\dots d_n \in \mathcal{D}(X)$ if there exists $\alpha_1, \dots, \alpha_n$ such that $\sum_{i}\alpha_i =1$ and  for all $x$, $d(x)=\sum_i \alpha_i d_i(x)$. Hereafter we will just write the latter condition as $d=\sum_i \alpha_i d_i$.

The \emph{convex closure} of a subset $S\subseteq \mathcal{D}(X)$, written $\conv(S)$, is the set of all the convex combinations of the distributions in $S$. A subset $S\subseteq \mathcal{D}(X)$ is called \emph{convex} if $S=\conv(S)$. A convex set is said to be finitely generated if there exists $d_1,\dots, d_n\in  \mathcal{D}(X)$ such that $S=\conv(\{d_1,\dots, d_n\})$. 

\medskip

This is enough to introduce the set of non-empty, finitely-generated convex sets of distributions, hereafter denoted as $C(X)$.

\medskip

A \emph{base} for $S\in C(X)$ is a set $\{d_1,\dots, d_n\}$ such that  $S=\conv(\{d_1,\dots, d_n\})$ and for all $i\in 1\dots n$, $d_i \notin \conv(\{ d_j \, |\, j\neq i, 1 \le j \le n \})$. 
	
\begin{theorem}\label{thm:uniquebase}
For every $S\in C(X)$, there exists a unique base.	
\end{theorem}

We present two proofs of this property. Proof I is based on functional analysis and the strong theorem of Krein-Milman~\cite{Rudin91}; Proof II is explicit and concrete. 

\medskip

\noindent{\bf Proof I.}
Let $S$ be an element of $C(X)$. Note that then $S$ is a subset of $\dset(X) \subseteq \mathbb{R}^X$ and hence a subset of a locally convex topological vector space ($\mathbb R^X$ with the product topology). 
Consider the family $$\mathcal B = \{B \subseteq S \mid S = \conv(B)\}.$$
It is obvious that $B$ is minimal in $\mathcal B$ if and only if no element $d \in B$ satisfies $d \in \conv(B\setminus\{d\})$. We are going to show that $\mathcal B$ contains a smallest element.

We first show that for all $B \in \mathcal B$, $\Ext(S) \subseteq B$ (*), with $\Ext(S)$ being the set of extreme points of $S$.

Indeed, let $d \in \Ext(S)$. Then $d \in S$ and can be written as $d = \sum_{d_i \in B} p_id_i = p_i\cdot d_i + (1-p)\cdot e$ for some $p_i \neq 0$ and $e \in S$, and hence by extremality of $d$ we have $d = d_i = e$ yielding $d \in B$. 

Next, we show that $S = \conv(\Ext(S))$, which means that $\Ext(S) \in \mathcal{B}$ and hence together with (*) shows that $\Ext(S)$ is the smallest element of $\mathcal B$. This smallest element $\Ext(S)$ is the unique base of $S$.  Pick a finite $B_0 = \{d_1, \dots, d_n\} \in \mathcal B$. Then $S = \Phi(\Delta_n)$ for $$\Delta_n = \{(x_1, \dots, x_n) \in \mathbb R^n \mid x_i \in [0,1], \sum_i x_i = 1\}$$ and 
$\Phi\colon \mathbb R^n \to X$ given by $\Phi(x_1, \dots, x_n) = \sum_i x_id_i $. Note that $\Delta_n$ is compact, by Heine-Borel, as it is a closed and bounded subset of $\mathbb R^n$, and $\Phi$ is continuous, since we are in a topological vector space and hence algebraic operations are continuous. 
As a consequence, $S$ is compact as a continuous image of a compact set. Now, Krein-Milmann applies, yielding that $S = \overline{\conv}(\Ext(S))$ with $\overline{\conv}$ denoting the closed convex hull and hence
$$ S = \overline{\conv}(\Ext(S)) = \conv(\Ext(S))$$ since by the same argument as above $\conv(\Ext(S))$ is compact and hence closed. \qed

Instead of the Krein-Milman theorem, one could use in this proof its predecessor from classical convex analysis in $\mathbb R^n$, e.g. \cite[Theorem 18.5]{Rockafellar72}. The reason is that since we deal with finitely generated convex subsets of finitely supported distributions, such subsets are actually elements of $CX$ for a finite set $X$. 

\medskip

\noindent{\bf Proof II.}
Existence of the base comes from the property that $S$ is finitely generated.
In the rest of this section we prove uniqueness; namely if $\{d_1,\dots, d_n\}$ and $\{d_1',\dots, d_m'\}$ are two bases for some $S\in \mathcal{D}(X)$, then $\{d_1,\dots, d_n\}=\{d_1',\dots, d_m'\}$, that is $n=m$ and there exists a permutation $\rho\colon n\to n$ such that $d_i'=\rho(d_i)$.

\medskip

Let  $\{d_1,\dots, d_n\}$ and $\{d_1',\dots, d_m'\}$ be two bases for $S\in \mathcal{D}(X)$. Then for all $i$,
$$d_i=\conv (\{d_1',\dots, d_m'\} ) \text{ and } d_i'=\conv(\{d_1,\dots, d_n\})\text{.}$$ 
By unfolding the definition of $\conv$, this just means that for all $i$ there exist $\alpha_{i,j}$ and $\alpha_{i,j}'$ such that $\sum_j \alpha_{i,j}=1$, $\sum_j \alpha_{i,j}'=1$, 
\begin{equation}\label{eq:proofofuniqueness}d_i=\sum_{j\in\{1\dots m\}} \alpha_{i,j }d_j' \text{ and } d_i'=\sum_{j\in\{1\dots n\}} \alpha_{i,j }'d_j\text{.}\end{equation}
By replacing $d_j'$ in the left equation in \eqref{eq:proofofuniqueness} with the one in the right we obtain that for all $i$
$$d_i=\sum_{j\in\{1\dots m\}} \alpha_{i,j }( \sum_{k\in\{1\dots n\}} \alpha_{j,k }'d_k)$$ 
This is equivalent to
$$d_i=\sum_{k\in\{1\dots n\}} (\sum_{j\in\{1\dots m\}} \alpha_{i,j }  \alpha_{j,k }') d_k$$ 
and thus 
\begin{equation}\label{eq:di}d_i= (\sum_{j\in\{1\dots m\}} \alpha_{i,j }   \alpha_{j,i }') d_i  \, +\, \sum_{k\in\{1\dots n\}\setminus \{i\}} (\sum_{j\in\{1\dots m\}} \alpha_{i,j }  \alpha_{j,k }') d_k\end{equation}

By reasoning in the same way, but replacing in the right equation of \eqref{eq:proofofuniqueness} the definitions of $d_i$ on the left, one obtains
\begin{equation}\label{eq:di'}d_i'= (\sum_{j\in\{1\dots n\}} \alpha_{i,j }'   \alpha_{j,i }) d_i'  \, +\, \sum_{k\in\{1\dots m\}\setminus \{i\}} (\sum_{j\in\{1\dots n\}} \alpha_{i,j }'  \alpha_{j,k }) d_k'\end{equation} 

\medskip

Now observe that all these equations are of the shape $e=\alpha e+ (1-\alpha)e_1$ for $\alpha\in[0,1]$ and $e,e_1\in \mathcal{D}(X)$. Whenever $e\neq e_1$, this kind of equation has $\alpha=1$ as unique solution. 
Now observe that in \eqref{eq:di}, we have that $d_i \neq \sum_{k\in\{1\dots n\}\setminus \{i\}} (\sum_{j\in\{1\dots m\}} \alpha_{i,j }  \alpha_{j,k }') d_k$ otherwise $d_i$ would be expressible as a convex combination of the others and therefore $\{d_1\dots, d_n\}$ would not be a base. 
Therefore we have that
\begin{equation}\label{eq:total1}\sum_{j\in\{1\dots m\}} \alpha_{i,j }   \alpha_{j,i }'  =1 \text{ for all }i \in \{1 \dots n\}\end{equation}
and that for all $k \in\{1\dots n\}\setminus \{i\}$, $\sum_{j\in\{1\dots m\}} \alpha_{i,j }  \alpha_{j,k }'=0$. Since all the summands are non-negative, this entails that
\begin{equation}\label{eq:10}
\alpha_{i,j }  \alpha_{j,k }'=0 \text{ for all }i \in \{1 \dots n\},\, k\in \{1\dots n\}\setminus \{i\} \text{ and } j\in\{1\dots m\}.
\end{equation}
By reasoning in the same way, we obtain from \eqref{eq:di'},
\begin{equation}\sum_{j\in\{1\dots n\}} \alpha_{i,j }'   \alpha_{j,i }  =1 \text{ for all }i \in \{1 \dots m\}\end{equation}
and 
\begin{equation}\label{eq:20}
\alpha_{i,j }'  \alpha_{j,k }=0 \text{ for all } i \in \{1 \dots m\},\, k\in \{1\dots m\}\setminus \{i\} \text{ and } j\in\{1\dots n\}.
\end{equation}

\medskip

We now prove that it must be that $n=m$ and there exists a permutation $\rho\colon\{1,\dots,n\} \to \{1,\dots,n\}$ such that $\alpha_{i,\rho(i)} = \alpha'_{\rho(i),i} = 1$ for all $i \in \{1,\dots,n\}$. 

\medskip

First we prove that for all $i$, there exists only one $j$ such that $\alpha_{i,j}=1$ and for all $k\neq j$, $\alpha_{i,k}=0$. Assume that there exists an $i,j,j'$ such that $\alpha_{i,j}=x$ and $\alpha_{i,j'}=x'$ with $x,x'\in (0,1]$. By \eqref{eq:10} one has that for all $k\neq i$, $\alpha_{j,k}'=0$ and $\alpha_{j',k}'=0$. Since $\sum_k \alpha_{j,k}'=1$ and $\sum_k \alpha_{j',k}'=1$, we have that $\alpha_{j,i}'=1$ and $\alpha_{j',i}'=1$. From \eqref{eq:20}, $\alpha_{i,l}=0$ for all $l\neq j$ and $\alpha_{i,l'}=0$
for all $l'\neq j'$. Therefore $j=j'$. This means that there exists only one $j$, such that $\alpha_{i,j}\neq 0$. Since $\sum_k \alpha_{i,k}=1$, we have that $\alpha_{i,j}=1$ and $\alpha_{i,k}=0$ for all $k\neq j$.

This defines a function $\rho\colon \{1,\dots,n\} \to \{1,\dots,m\}$ mapping each $i$ into the unique $j$ such that $\alpha_{i,j}=1$. By the same reasoning we can define a function $\rho'\colon \{1,\dots,m\} \to \{1,\dots,n\}$ mapping each $i$ in the only $j$ such that $\alpha_{i,j}'=1$.

\medskip

We conclude by showing that $\rho'$ must be the inverse of $\rho$.
Assume that  $\alpha_{i,j}=1$. For what we have proved so far, $\alpha_{i,k}=0$ for all $k\neq j$. By \eqref{eq:total1}, $\alpha_{j,i}'=1$. \qed

\section{Monads and presentations}\label{sec:monads}
Theorem~\ref{thm:uniquebase} states the existence of a unique base for every convex subset of probability distributions. 
In the remainder of this paper, we exploit this result to illustrate an alternative proof of Theorem 4 in~\cite{BSV19arxiv} that provides a presentation of the monad $C$~\cite{Mislove00,Goubault-Larrecq08a,TixKP09a, Varacca03,VaraccaW06,Jacobs08}. In Section~\ref{sec:monadC}, we recall the monad as well as its presentation given in~\cite{BSV19arxiv}. In this section, we recall some basic facts about monads  and presentations.

\medskip

A \emph{monad}  on $\Sets$ is
a functor $\T\colon\Sets \rightarrow \Sets$ together with two
natural transformations: a unit $\eta\colon \Idmap
\Rightarrow \T$ and multiplication $\mu \colon \T^{2} \Rightarrow
\T$ that satisfy the laws $\mu \after \eta\T = \mu \after \T\eta = \id$ and $\mu\after \T\mu = \mu \after\mu\T$.

A \emph{monad map} from a monad $\T$ to a monad $\hat\T$ is a natural transformation $\sigma\colon \T\Rightarrow\hat\T$ that makes the following diagrams commute, with $\eta, \mu$ and $\hat\eta, \hat\mu$ denoting the unit and multiplication of $\T$ and $\hat\T$, respectively, and $\sigma\sigma = \sigma\after \T \sigma = \hat\T \sigma \after \sigma_\T$.
\vskip-0.5cm
$$\xymatrix@R-1pc{
&{X}\ar[dr]_{\hat\eta}\ar[r]^-{\eta} & \T X\ar[d]^{\sigma}
&  
{\T\T X}\ar[d]_{\mu}\ar[r]^-{\sigma\sigma} & {\hat\T\hat\T X}\ar[d]^{\hat\mu} &
\\
&& {\hat\T X} 
& 
{\T X}\ar[r]_-{\sigma} & {\hat\T X}&
}$$
\vskip-0.2cm
If $\sigma\colon \T X \to \hat\T X$ is an epi monad map, then $\hat\T$ is a \emph{quotient} of $\T$. If it is a mono, then $\T$ is a \emph{submonad} of $\hat\T$. If it is an iso, the two monads are isomorphic.

\medskip

An important example of monad is provided by the \emph{free monad of terms}. Given a signature $\Sigma$, namely a set of operation symbols equipped with an arity,  the free monad $T_\Sigma\colon \Sets \to \Sets$ of terms over $\Sigma$ maps a set $X$ to the set of all $\Sigma$-terms with variables in $X$, and $f\colon X \to Y$ to the function that maps a term over $X$ to a term over $Y$ obtained by substitution according to $f$. The unit maps a variable in $X$ to itself, and the multiplication is term composition. 

Given a set of axioms $E$ over $\Sigma$-terms, one can define the smallest congruence generated by the axioms, denoted by $=_E$. Hereafter we write $[t]_E$ for the $=_E$-equivalence class of the $\Sigma$-term $t$ and $T_{\Sigma,E}(X)$ for the set of $E$-equivalence classes of $\Sigma$-terms with variables in $X$. The assignment $X \mapsto T_{\Sigma,E}(X)$ gives rise to a functor  $T_{\Sigma,E} \colon \Sets \to \Sets$ where the behaviour on functions is defined as for $T_\Sigma$. Such functor carries the structure of a monad: the unit $\eta^E\colon \Idmap \Rightarrow T_{\Sigma,E}$ and the multiplication $\mu^E \colon T_{\Sigma,E}T_{\Sigma,E}\Rightarrow T_{\Sigma,E}$ are defined as $\eta^{E} (x)= [x]_{E}$ and 
$\mu^{E} [ t \{[t_{i}]_{E}/ x_{i}\}]_{E}= [t \{t_{i}/ x_{i}\}]_{E}$.
%
%

%

An \emph{algebraic theory} is a pair $(\Sigma, E)$ of signature $\Sigma$ and a set of equations $E$.  We say that $(\Sigma,E)$ provides a \emph{presentation} for a monad $\T$ if $T_{\Sigma,E}$ is isomorphic to $\T$.

We next introduce several monads on $\Sets$ together with their presentations.

\medskip

\noindent{\bf Nondeterminism.} The non-empty finite powerset monad $\Pow_{ne}$ maps a set $X$ to the set of non-empty finite subsets $\Pow_{ne} X = \{U \mid U \subseteq X, \,\, U \textrm{ is finite and non-empty}\}$ and a function $f\colon X\to Y$ to $\Pow_{ne} f\colon \Pow_{ne} X\to \Pow_{ne} Y$, $\Pow_{ne} f (U) = \{f(u) \mid u\in U\}$. 
The unit $\eta$ of $\Pow_{ne}$ is given by singleton, i.e., $\eta(x) = \{x\}$ and the multiplication $\mu$ is given by union, i.e., $\mu(S) = \bigcup_{U \in S} U$ for $S \in \Pow_{ne}\Pow_{ne} X$. 

Let $\Sigma_N$ be the signature consisting of a binary operation $\nplus $.
Let $E_N$ be the following set of axioms, the axioms of semilattice:
$$\begin{array}{ccc}
(x\nplus y)\nplus z& \stackrel{(A)}{=}& x\nplus (y\nplus z) \\
x\nplus y& \stackrel{(C)}{=}& y\nplus x\\
x\nplus x&\stackrel{(I)}{=}&x
\end{array}$$
It is easy to show that the algebraic theory $(\Sigma_N,E_N)$ provides a presentation for the monad $\nePow$, in the sense that there exists an isomorphism of monads $\iota^N\colon T_{\Sigma_N,E_N}\Rightarrow \nePow$.

\medskip

\noindent{\bf Probability.} The finitely supported probability distribution monad $\Dis$ is defined, for a set $X$ and a function $f\colon X \to Y$,  as
\vskip-0.7cm
\begin{align*}
&\Dis X
 = 
\set{\varphi\colon X \to [0,1]}{\sum_{x \in X} \varphi(x) = 1,\, \supp(\varphi) \text{~is~finite}}\\
&\Dis f(\varphi)(y)
 = 
\sum\limits_{x\in f^{-1}(y)} \varphi(x).	
\end{align*}
\vskip-0.2cm
The unit of $\Dis$ is given by a Dirac
distribution $\eta(x) = \delta_x = ( x \mapsto 1)$ for $x \in X$ and
the multiplication by $\mu(\Phi)(x) = \sum_{\varphi \in \supp(\Phi)}
\Phi(\varphi)\cdot \varphi(x)$ for $\Phi \in \Dis\Dis X$.
\noindent We sometimes write $\sum_{i \in I} p_i x_i$ for a distribution $\varphi$ with $\supp(\varphi) = \{x_i \mid i\in I\}$ and $\varphi(x_i) = p_i$.

Let $\Sigma_P$ be the signature consisting of a binary operation $+_p$ for all $p\in (0,1)$.
Let $E_P$ be the following set of axioms, the axioms of a barycentric algebra also called convex algebra~:\footnote{There is another equivalent presentation for convex algebras with a signature involving arbitrary convex combinations and two axioms, projection and barycenter. In this paper we will mainly use the binary convex operations.}
$$\begin{array}{ccc}
(x+_qy)+_pz& \stackrel{(A_p)}{=}& x+_{pq}(y+_{\frac{p(1-q)}{1-pq}}z)\\
x+_py& \stackrel{(C_p)}{=} & y+_{1-p}x\\
x+_px&\stackrel{(I_p)}{=}&x
\end{array}$$
The algebraic theory $(\Sigma_P,E_P)$ provides a presentation for the monad $\mathcal{D}$~\cite{swirszcz:1974,semadeni:1973,doberkat:2006,doberkat:2008,jacobs:2010}, in the sense that there exists an isomorphism of monads $\iota^P\colon T_{\Sigma_P,E_P}\Rightarrow \mathcal{D}$.


\subsection{A well known recipe for constructing monad morphisms}\label{sec:recipe}
To prove that an algebraic theory $(\Sigma,E)$ presents a monad $\T$, one has to provide $\iota \colon T_{\Sigma,E}\Rightarrow \T$ that (a) is a monad map and (b) is an isomorphism. While the proof of (b) often require some ad-hoc normal form arguments, the proof of (a) can be significantly simplified   by using some some standard categorical machinery.

In this section, we illustrate a well known recipe which allows for constructing a monad map $\iota\colon T_{\Sigma,E} \Rightarrow \T$ in a principled way.  We begin by recalling Eilenberg-Moore algebras.

\medskip

To each  monad $\T$, one associates the  Eilenberg-Moore category
$\EM(\T)$ of $\T$-algebras. Objects of
$\EM(\T)$ are pairs $\Alg A = (A, a)$ of a set $A \in \Sets$ and a map
$a\colon \T A \rightarrow A$,
making the first two
diagrams below commute. 
\vskip-0.3cm
$$\xymatrix@R-1pc@C-1pc{
A\ar@{=}[dr]\ar[r]^-{\eta} & \T A\ar[d]^{a}
& 
\T^{2}A\ar[d]_{\mu}\ar[r]^-{\T a} & \T A\ar[d]^{a}
& 
\T A\ar[d]_{a}\ar[r]^-{\T h} & \T B\ar[d]^{b} \\
& A 
&
\T A\ar[r]_-{a} & A
 &
A\ar[r]_-{h} & B
}$$
\vskip-0.2cm
\noindent A homomorphism from an algebra $\Alg A = (A, a)$ to an algebra $\Alg B = (B, b)$ is a map $h\colon
A\rightarrow B$ between the underlying sets making the third diagram above commute.

It is well known that, when $\T$ is the monad $T_{\Sigma,E}$ for some algebraic theory $(\Sigma,E)$, $\EM(\T)$ is isomorphic to the category $\AlgCat(\Sigma,E)$ of $(\Sigma,E)$-algebras and their morphisms. 
A $\Sigma$-algebra $(X, \Sigma_X)$ consist of a set $X$ together with a set $\Sigma_X$ of operations $\hat{o}_X\colon X^n \to X$, one for each operation symbol $o \in \Sigma$ of arity $n$. A $(\Sigma,E)$-algebra is a $\Sigma$-algebra where all the equations in $E$ hold. A homomorphism $h$ from a $(\Sigma,E)$-algebra $(X, \Sigma_X)$ to a $(\Sigma,E)$-algebra $(Y, \Sigma_Y)$ is a function $h\colon X \to Y$ that commutes with the operations, i.e., $h \after \hat{o}_X = \hat{o}_Y \after h^n$ for all $n$-ary $o \in \Sigma$.

For instance, $(\Sigma_N,E_N)$-algebras are semilattices, namely a set $X$ equipped with a binary operation $\hat{\nplus}_X$ that is associative, commutative and idempotent. A semilattice homomorphism is a function $h\colon X \to Y$ such that $h(x_1 \hat{\nplus}_X x_2)= h(x_1) \hat{\nplus}_Y h(x_2)$ for all $x_1,x_2\in X$.

\medskip

Now we can display an abstract recipe for constructing a monad map $\iota\colon T_{\Sigma,E} \Rightarrow \T$, which consists of three steps:
\begin{itemize}[leftmargin=20pt]
\item[(A)] For each set $X$, provide $\T X$ with the structure of a $(\Sigma,E)$-algebra, namely functions $\hat{o}_X \colon (\T X)^n \to \T X$ for each $o\in \Sigma$, that satisfy the equations in $E$;
\item[(B)] Prove that for each function $f\colon X \to Y$, $\T f$ is a $(\Sigma,E)$-algebra homomorphism;
\item[(C)] Prove that for each set $X$, $\mu^{\T}_X\colon \T\T X \to \T X$ is a $(\Sigma,E)$-algebra homomorphism.
\end{itemize}

By the correspondence of $(\Sigma,E)$-algebras and Eilenberg-Moore algebra for $T_{\Sigma,E}$ and (A), we obtain a $T_{\Sigma,E}$-algebra $\alpha^{\sharp}_X\colon T_{\Sigma,E} \T X \to \T X$ for each set $X$. These $\alpha^{\sharp}_X$ give rise to a natural transformation $\alpha^{\sharp}\colon T_{\Sigma,E} \T \Rightarrow \T$ by (B) and the correspondence of  $(\Sigma,E)$-homomorphisms and  $T_{\Sigma,E}$-homomorphisms. The monad morphism $\iota\colon T_{\Sigma,E} \Rightarrow \T$ is then obtained by (C) and the following theorem\footnote{This theorem is known, but it is not easy to find an original reference for it. We thank Jurriaan Rot for recalling the theorem and the proof with us.}.

\begin{theorem} Let $(\T, \eta^\T, \mu^{\T})$ and $(\hat{\T},\eta^{\hat{\T}},\mu^{\hat{\T}})$ be two monads. Let  $\alpha^{\sharp}\colon {\T} \hat{\T} \Rightarrow \hat{\T}$ be a natural transformation such that $\alpha^{\sharp}_X \colon {\T} \hat{\T}X \to \hat{\T}X$ is an  Eilenberg-Moore algebra for ${\T}$ and that $\mu_X^{\hat{\T}} \colon \hat{\T}\hat{\T}X \to \hat{\T}X$ is an ${\T}$-algebra morphism from $(\hat{\T}\hat{\T}X, \alpha^{\sharp}_{\hat{\T}X})$ to $(\hat{\T}X, \alpha^{\sharp}_X)$. Then $$\iota:= \xymatrix{{\T} \ar@{=>}[r]^{{\T}\eta^{\hat{\T}}} & {\T}\hat{\T} \ar@{=>}[r]^{\alpha^{\sharp}} & \hat{\T}}$$ is a monad map.
\end{theorem}
\begin{proof}
In order to prove that $\iota$ is a monad map, we need to prove that the following two diagrams commute.
\begin{equation}\label{eq:proofmorphism}
\xymatrix{
X \ar[rd]_{\eta^{\hat{\T}}}\ar[r]^{\eta^\T}& {\T}X \ar[d]^{\iota_X}\\
& {\hat{\T}}X }
\qquad
\xymatrix{
{\T}{\T}X \ar[d]_{\mu^E}\ar[r]^{{\T}\iota_X}& {\T}{\hat{\T}}X \ar[r]^{\iota_{{\hat{\T}}X}}
& {\hat{\T}}{\hat{\T}}X \ar[d]^{\mu^{\hat{\T}}}\\
{\T}X \ar[rr]_{\iota_X}&& {\hat{\T}}X
}
\end{equation}

For proving commutation of the diagram on the left, it is enough to recall that $\iota = \alpha^{\sharp} \after {\T}\eta^{{\hat{\T}}}$ and observe that the following diagram commutes: the top square commutes by naturality of $\eta^{\T}$ and the bottom triangle commutes since $\alpha^{\sharp}_X$ is an Eilenberg Moore algebra for ${\T}$.
$$
\xymatrix{
X \ar[r]^{\eta^{\T}_X}\ar[d]_{\eta^{{\hat{\T}}}_X}& {\T}X\ar[d]^{{\T}\eta^{{\hat{\T}}}} \\
{\hat{\T}}X \ar[rd]_{id} \ar[r]^{\eta^{\T}_{{\hat{\T}}X}}& {\T} {{\hat{\T}}}X\ar[d]^{\alpha^\sharp_X}\\
& {{\hat{\T}}}X
}$$

In order to prove the commutation of the diagram on the right in \eqref{eq:proofmorphism}, we need the assumption that $\mu_X^{\hat{\T}}$ is a homomorphism of ${\T}$-algebras, namely that the following diagram commutes.
\begin{equation}\label{eqref:mumorphism}
\xymatrix{
{\T}{\hat{\T}}{\hat{\T}}X \ar[d]_{\alpha^{\sharp}_{{\hat{\T}}X}} \ar[r]^{{\T}\mu_X^{{\hat{\T}}}} & {\T} {{\hat{\T}}}X \ar[d]^{\alpha_X^{\sharp}} \\
{{\hat{\T}}}{{\hat{\T}}}X \ar[r]_{\mu_X^{\hat{\T}}} & {\hat{\T}}X
}
\end{equation}

By recalling that $\iota = \alpha^{\sharp} \after {\T}\eta^{{\hat{\T}}}$, the left triangle below commutes and, since \eqref{eqref:mumorphism} commutes, the whole following diagram commutes.
$$\xymatrix{
{\T}{\hat{\T}}X \ar[r]^{{\T}\eta_{{\hat{\T}}X}^{\hat{\T}}} \ar[rd]_{\iota_{{\hat{\T}}X}} &{\T}{\hat{\T}}{\hat{\T}}X \ar[d]^{\alpha^{\sharp}_{{\hat{\T}}X}} \ar[r]^{{\T}\mu_X^{\hat{\T}}} & {\T}{\hat{\T}}X \ar[d]^{\alpha_X^{\sharp}} \\
& {\hat{\T}}{\hat{\T}}X \ar[r]_{\mu_X^{\hat{\T}}} & {\hat{\T}}X
}
$$
Since ${\T}\mu_X^{{\hat{\T}}} \after {\T}\eta_{{{\hat{\T}}}X}^{{\hat{\T}}} = {\T}(\mu_X^{\hat{\T}} \after \eta_{{\hat{\T}}X}^{\hat{\T}}) ={\T}(id_{{\hat{\T}}X})=id_{{\hat{\T}}X}$, we have $$\alpha^{\sharp}_X = \mu_X \after \iota_{{\hat{\T}}X}\text{.}$$

So, proving that the right diagram in \eqref{eq:proofmorphism} commutes, amounts to proving that the following diagram commutes.
$$\xymatrix{
{\T}{\T}X \ar[d]_{\mu^E}\ar[r]^{{\T}\iota_X}& {\T}{\hat{\T}}X \ar[d]^{\alpha^{\sharp}_X} \\
{\T}X \ar[r]_{\iota_X}& {\hat{\T}}X
}$$
By recalling that $\iota = \alpha^{\sharp} \after {\T}\eta^{\hat{\T}}$ it is equivalent to prove that the following commutes
$$\xymatrix{
{\T}{\T}X \ar[d]_{\mu^{\T}_X}\ar[r]^{{\T}{\T}\eta_X^{\hat{\T}}}& {\T}{\T}{\hat{\T}}X \ar[r]^{{\T}\alpha^{\sharp}_X} \ar[d]_{\mu_{{\hat{\T}}X}^{\T}} & {\T}{\hat{\T}}X \ar[d]^{\alpha^{\sharp}_X} \\
{\T}X \ar[r]_{\T \eta_X^{\hat{\T}}}&{\T}{\hat{\T}}X \ar[r]_{\alpha^{\sharp}_X}& {\hat{\T}}X
}$$
The left square commutes by naturality of $\mu^{\T}$. The right square commutes since $\alpha^{\sharp}_X$ is an Eilenberg-Moore algebra for ${\T}$.
\end{proof}

The function $\iota_X\colon T_{\Sigma,E}X \to \T X$ obtained by the above recipe can be inductively defined for all $x\in X$, $t_1,\dots, t_n \in T_{\Sigma}X$ and $n$-ary operations $o$ in $\Sigma$ as follows.
\begin{equation}
\iota_X([x]_E) = \eta_X^{\T}(x) \qquad \iota_X([o(t_1, \dots, t_n)]_E) = \hat{o}_X (\iota_X[t_1]_E, \dots, \iota_X[t_n]_E)\text{.}
\end{equation}
The fact that the functions $\hat{o}_X$ form a $(\Sigma,E)$-algebra ensures that $\iota$ is a well defined function, namely if $t=_E t'$, then $\iota([t]_E)=\iota([t']_E)$.

\medskip

We conclude this section by shortly illustrating how to apply the above recipe to the monad for non-determinism and the one for probability discussed above. To construct a monad map $\iota^N\colon T_{\Sigma_N,E_N}\Rightarrow \Powne$, we define for all sets $X$ the binary function $\hat{\nplus}\colon \Powne(X)\times \Powne(X) \to \Powne(X)$ as the union $\cup$. This is associative, commutative and idempotent, so the axioms in $E_N$ are satisfied, or in other words, this forms a semilattice. This corresponds to point (A) of the recipe. It is not difficult to check (B) and (C). The resulting monad map is defined for all sets $X$ as 
\begin{equation*}
\iota_X^N([x]_{E_N}) = \{x\} \qquad \iota_X^N([t_1 \nplus t_2]_{E_N}) = \iota_X^N([t_1]_{E_N}) \cup \iota_X^N([t_2]_{E_N})\text{.}
\end{equation*}
To construct the monad map $\iota^P\colon T_{\Sigma_P,E_P}\Rightarrow \Dis$,  we define for all $p\in (0,1)$ and all sets $X$ the binary function $\hat{+}_p\colon \Dis(X)\times \Dis(X) \to \Dis(X)$ as 
\begin{equation}
d_1\hat{+}_pd_2 = pd_1 +(1-p)d_2\text{.}
\end{equation}
One can check that the three axioms in $E_P$ are satisfied (distributions form a famous convex algebra), and that points (B) and (C) of the recipe hold. The resulting monad map is defined for all sets $X$ as 
\begin{equation}\label{eq:iotap}
\iota_X^P([x]_{E_P}) = \delta_x \qquad \iota_X^P([t_1 +_p t_2]_{E_P}) = p\iota_X^P([t_1]_{E_P}) + (1-p) \iota_X^P([t_2]_{E_P})\text{.}
\end{equation}

\section{The monad for non-determinism and probability}\label{sec:monadC}
In this section, we recall the monad for non-determinism and probability, its presentation, and we illustrate some interesting properties.

The monad  $C\colon \Sets \to \Sets$ maps a set $X$ into $CX$, namely the set of non-empty, finitely-generated convex subsets of distributions on $X$ (as defined in Section \ref{sec:ubthm}). 
For a function $f \colon X \to Y$, $Cf\colon CX \to CY$ is given by 
$Cf(S) = \{\Dis f(d) \mid d \in S\}$.
The unit of $C$ is $\eta\colon X \to CX$ given by $\eta(x) = \{ \delta_x\}$.
The multiplication of $C$, $\mu\colon CCX \to CX$ can be expressed in concrete terms as follows~\cite{Jacobs08}. Given $S\in CCX$,
$$\mu (S) = \bigcup_{\Phi \in S} \{ \sum_{U \in \supp \Phi} \Phi(U)\cdot d \mid d \in U\}.$$ 

Let $\Sigma$ be the signature $\Sigma_N \cup \Sigma_P$. Let $E$ be the sets of axioms consisting of $E_N$, $E_p$ and the following distributivity axiom:

$$(x\nplus y) +_p z\stackrel{(D)}{=}(x+_p z) \nplus (y +_p z)
$$
This theory $(\Sigma, E)$ is the algebraic theory of \emph{convex semilattices}, introduced in~\cite{BSV19arxiv}.

\begin{theorem}\label{thm:presentation}
$(\Sigma,E)$ is a presentation of the monad $C$.
\end{theorem}
The above theorem has been proved in~\cite{BSV19arxiv}. In the remainder of this paper, we will provide an alternative proof of this fact by exploiting the unique base theorem (Theorem~\ref{thm:uniquebase}).

We begin by observing that the assignment $S \mapsto \conv(S)$ gives rise to a natural transformation, that we refer hereafter as 
\begin{equation}\label{eq:conv}
\conv \colon \Powne \Dis \Rightarrow C.
\end{equation}
Theorem~\ref{thm:uniquebase} provides a way of going backward, from $C$ to $\Powne \Dis$: we call  $UB_X\colon CX \to \Powne \Dis X$ the function assigning to each convex subset $S$ its unique base. However such $UB_X$ does not give rise to a natural transformation, in the sense that the following diagram does not commute 
\begin{equation}\label{eq:notnatural}\xymatrix{CX \ar[d]_{UB_X} \ar[r]^{C f} & CY \ar[d]^{UB_Y} \\
\Powne \Dis X \ar[r]_{ \Powne \Dis f}&  \Powne \Dis Y}
\end{equation}
for arbitrary function $f\colon X \to Y$. It holds that $UB_Y \after Cf \subseteq \Powne \Dis f \after UB_X$ but not the other way around, as shown by the next example.

\begin{example}
Let $X=\{x,y,z\}$, $Y=\{a,b\}$ and $f\colon X \to Y$ be the  function mapping both $x$ and $y$ to $a$ and $z$ to $b$. Consider the set $S=\{\frac{1}{2}x+\frac{1}{2}y, \, \frac{1}{2}x+\frac{1}{2}z,\, \delta_z \}$: this set is a base since none of its element can be expressed as convex combination of the others. However, the set $ \Powne \Dis f (S)=\{\delta_a, \, \frac{1}{2}a+\frac{1}{2}b, \, \delta_b \}$ is not a base since $\frac{1}{2}a+\frac{1}{2}b$ can be expressed as linear combination of $\delta_a$ and $\delta_b$. Now, by taking the convex set $\conv (S) \in CX$ one can easily see that $UB_Y \after Cf \not \supseteq \Powne \Dis f \after UB_X$. Indeed $\Powne \Dis f \after UB_X (\conv (S)) = \Powne \Dis f (S)= \{\delta_a, \, \frac{1}{2}a+\frac{1}{2}b, \, \delta_b \}$, while $UB_Y \after Cf (\conv(S))= \{\delta_a,\, \delta_b\}$ since $Cf (\conv(S)) = \conv(\Dis f(S))$ by Lemma~\ref{lem:ubf} below.
\end{example}

Interestingly enough, while \eqref{eq:notnatural} does not commute, the following diagram  does.
\begin{equation*}
\xymatrix{CX \ar[d]_{UB_X} \ar[r]^{C f} & CY  \\
 \Powne \Dis X \ar[r]_{ \Powne \Dis f}&  \Powne \Dis Y \ar[u]_{\conv_Y}}
\end{equation*}

This is closely related to Lemma 37 from~\cite{BSV19arxiv} which provides a slightly different formulation.
Below, we illustrate a proof: to simplify the notation of the natural transformations $\conv_X$ and $UB_X$ we avoid to specify the set $X$ whenever it is clear from the context.

\begin{lemma}\label{lem:ubf}
Let $S\in \cset(X)$ and $f:X\to Y$. Then 
$\cset f(S) 
= \conv (\{\dset f (d)\mid d \in UB(S)\})$.
\end{lemma}
\begin{proof}
We prove $\cset f(S) \subseteq \conv \bigcup_{\done \in UB(S)} \{\dset f (\done)\}$.
Let $\dtwo \in \cset f (S)$. 
Then $\dtwo= \dset f (\done)$ for some $\done \in S$, 
which implies that $\done$ is a convex combination of elements of $UB(S)$, that is,
$\done=\sum_{i} p_{i} \cdot \done_{i}$ with $\done_{i} \in UB(S)$ for all $i$.
Hence, $\dtwo= \sum_{i} p_{i} \cdot \dset f(\done_{i}) \in \conv \bigcup_{\done \in UB(S)} \{\dset f (\done)\}$.

For the opposite inclusion, let $\dtwo\in \conv \bigcup_{\done \in UB(S)} \
[\dset f (\done)\
]$. 
Hence, $\dtwo= \sum_{i} p_{i} \cdot \dset f(\done_{i})$ with $\done_{i} \in UB(S)$ for all $i$. 
We have $\sum_{i} p_{i} \cdot \dset f(\done_{i})=\dset f (\sum_{i} p_{i} \cdot \done_{i})$ and, 
by $\sum_{i} p_{i} \cdot \done_{i} \in S$, 
we conclude $\dtwo \in \cset f(S)$.

\end{proof}

%

\section{The monad map $\iota\colon T_{\Sigma,E} \Rightarrow C$}\label{secproofOne}

In this section we apply the standard recipe from Section \ref{sec:recipe} to construct a monad map $\iota\colon T_{\Sigma,E} \Rightarrow C$.

For this aim, we first recall two well-known operations on convex sets:
the convex union $\oplus \colon C(X) \times C(X) \to C(X)$ defined for all $S_1,S_2\in C(X)$ as $$S_1\oplus S_2=\conv(S_1 \cup S_2)$$
and, for all $p \in (0,1)$, the Minkowski sum $\oplus_p \colon C(X) \times C(X) \to C(X)$ defined as
$$S_1 \mplus_p S_2=\{d \,|\, d=pd_1+(1-p)d_2 \text{ for some } d_1\in S_1 \text{ and } d_2\in S_2\}\text{.}$$

Point (A) and (B) of the recipe are guaranteed by the following small result from~\cite[Lemma 38]{BSV19arxiv}.

\begin{lemma}\label{lem:CS-conv-semilattice}
	With the above defined operations $(CX, \oplus, +_p)$ is a convex semilattice. Moreover, for a map $f\colon X \to Y$, the map $Cf\colon CX \to CY$ is a convex semilattice homomorphism from $(CX, \oplus, +_p)$ to $(CY, \oplus, +_p)$. \qed
\end{lemma}

The following lemma proves point (C) explicitly, namely that $\mu$ is a $(\Sigma,E)$-homomorphism. Note that this is already (implicitly) proven in~\cite{BSV19arxiv}: There we first note that $(CX, \oplus, +_{p})$ is the free convex semilattice generated by $X$ and then prove that $\mu = \id_{CX}^\#$, see~\cite[Lemma 41]{BSV19arxiv}, which means that $\mu$ is the unique homomorphism (and hence certainly a homomorphism) from the free convex semilattice generated by $CX$ to the free convex semilattice generated by $X$ that extends the identity map on $CX$. 

\begin{lemma}\label{lem:muhom}
For all $S_1,S_2\in CC(X)$, it holds that:
\begin{enumerate}
\item $\mu (S_1 \oplus S_2)= \mu(S_1) \oplus \mu(S_2)$
\item $\mu (S_1 \mplus_p S_2)= \mu(S_1) \mplus_p \mu(S_2)$
\end{enumerate}
\end{lemma}
\begin{proof}

\newcommand{\defi}{\mathrel{\overset{\makebox[0pt]{\mbox{\normalfont\tiny\sffamily def}}}{=}}}


Through this proof, we will often use the following key observation: $d\in \mu(S)$ iff
\begin{equation*}
\exists \Phi \in S \text{ such that } d=\sum_{U\in supp(\Phi)}\Phi(U)\cdot f(U) \text{ , for } f\colon supp(\Phi)\to \Dis(X) \text{ such that } f(U)\in U\text{.}
\end{equation*}

\begin{enumerate}
\item We first prove the inclusion $\mu(S_1) \oplus \mu(S_2)\subseteq \mu (S_1 \oplus S_2)$.

As $S_{1}\subseteq S_1 \oplus S_2$ we derive that
\begin{equation}\label{eq:mu11}
\mu(S_{1})\defi\bigcup_{\Phi \in S_{1}} \{\sum_{U \in \supp (\Phi)} \Phi(U) \cdot d \mid d\in U\} 
\subseteq \bigcup_{\Phi \in S_{1} \oplus S_{2}} \{\sum_{U \in \supp (\Phi)} \Phi(U) \cdot d \mid d\in U\}\defi \mu(S_{1} \oplus S_{2})
\end{equation}
Symmetrically, by $S_{2}\subseteq S_1 \oplus_p S_2$ we have
\begin{equation}\label{eq:mu12}
\mu(S_{2})\defi\bigcup_{\Phi \in S_{2}} \{\sum_{U \in \supp (\Phi)} \Phi(U) \cdot d \mid d\in U\} 
\subseteq \bigcup_{\Phi \in S_{1} \oplus S_{2}} \{\sum_{U \in \supp (\Phi)} \Phi(U) \cdot d \mid d\in U\} \defi \mu(S_{1} \oplus S_{2})
\end{equation}

Hence, 
\begin{align*}
\mu(S_1) \oplus \mu(S_2) &=\conv(\bigcup_{\Phi \in S_{1}} \{\sum_{U \in \supp (\Phi)} \Phi(U) \cdot d \mid d\in U\}  \cup \bigcup_{\Phi \in S_{2}} \{\sum_{U \in \supp (\Phi)} \Phi(U) \cdot d \mid d\in U\})\\
&\subseteq \conv(\mu(S_{1} \oplus S_{2})) \tag{by \ref{eq:mu11}, \ref{eq:mu12}}\\
&= \mu(S_{1} \oplus S_{2}) \tag{by $\mu(S_{1} \oplus S_{2})$ a convex set}
\end{align*}

We then prove the inclusion $\mu (S_1 \oplus S_2)\subseteq \mu(S_1) \oplus \mu(S_2)$.

Take $d\in \mu(S_1 \oplus S_2)$. Then there is a $\Phi\in S_1 \oplus S_2$ such that
$d=\sum_{U \in \supp (\Phi)} \Phi(U)\cdot f(U)$, with $f: \supp(\Phi) \to \dset(X)$ a function such that $f(U)\in U$.
As $\Phi$ is a convex combination of $(S_1 \cup S_2)$, we have $\Phi=\sum_{i} p_{i} \cdot \Phi_{i}$ with $\Phi_{i}\in (S_1 \cup S_2)$ for all $i$. 
Then for all $x\in X$ we have
\begin{align*}
\sum_{U \in \supp (\Phi)} \Phi(U)\cdot f(U)(x)
&= \sum_{U \in \cup_{i}\supp (\Phi_{i})} \big((\sum_{i} p_{i} \cdot \Phi_{i})(U)\cdot f(U)(x)\big)\\
&= \sum_{U \in \cup_{i}\supp (\Phi_{i})} \big(\sum_{i} p_{i} \cdot \Phi_{i}(U)\cdot f(U)(x)\big)\\
&= \sum_{i} p_{i} \cdot \big(\sum_{U \in \cup_{i}\supp (\Phi_{i})} \Phi_{i}(U)\cdot f(U)(x)\big)\\
&= \sum_{i} p_{i} \cdot \big(\sum_{U \in \supp (\Phi_{i})} \Phi_{i}(U)\cdot f(U)(x)\big)
\end{align*}
Hence, the result follows as 
\begin{align*}
d&=\sum_{i} p_{i} \cdot \big(\sum_{U \in \supp (\Phi_{i})} \Phi_{i}(U)\cdot f(U)\big)\\
&\in \conv (\bigcup_{\Phi \in (S_{1}\cup S_{2})} \{\sum_{U \in \supp (\Phi)} \Phi(U) \cdot d \mid d\in U\})\\
&= \mu(S_{1}\oplus S_{2})
\end{align*}

\item We first prove $ \mu(S_1) \mplus_p \mu(S_2) \subseteq \mu(S_1 \mplus_p S_2)$

Let $d\in \mu(S_1) \mplus_p \mu(S_2)$.
Then $d=
(\sum_{U \in \supp (\Phi_{1})} \Phi_{1}(U) \cdot f(U)) +_{p} (\sum_{U \in \supp (\Phi_{2})} \Phi_{2}(U) \cdot g(U))$ 
with $\Phi_{1} \in S_{1}, \Phi_{2} \in S_2$,  with $f: \supp(\Phi_{1}) \to \dset (X)$ such that $f(U)\in U$,  and with $g: \supp(\Phi_{2}) \to \dset (X)$ such that $g(U)\in U$.

We have that $d$ is equal to the probability distribution:
\[\sum_{U \in \supp (\Phi_{1}+_{p} \Phi_{2})} \big((\Phi_{1}+_{p} \Phi_{2})(U)\cdot h(U)\big)\]
with $h: \supp (\Phi_{1}+_{p} \Phi_{2}) \to \dset (X)$ defined as follows:
\[
h(U)=\begin{cases}
f(U) &\text{if $U \in (\supp (\Phi_{1})\setminus \supp(\Phi_{2}))$}\\
g(U) &\text{if $U \in (\supp (\Phi_{2})\setminus \supp(\Phi_{1}))$}\\
(f(U) +_{\frac{p\cdot \Phi_{1}(U)} {(\Phi_{1} +_{p} \Phi_{2})(U)}} g(U))
& \text{if $U\in (\supp (\Phi_{1})\cap \supp(\Phi_{2}))$}
\end{cases}
\]
To see this, take an $x\in X$. We have
\begin{align*}
d(x)&=\Big(\big(\sum_{U \in \supp (\Phi_{1})} \Phi_{1}(U) \cdot f(U)\big) +_{p} \big(\sum_{U \in \supp (\Phi_{2})} \Phi_{2}(U) \cdot g(U)\big) \Big) (x)\\
&=\big(\sum_{U \in \supp (\Phi_{1})} (p\cdot \Phi_{1}(U) \cdot f(U)(x))\big) + \big(\sum_{U \in \supp (\Phi_{2})} ((1-p)\cdot \Phi_{2}(U) \cdot g(U)(x))\big)\\
&=\big(\sum_{U \in \supp (\Phi_{1})\setminus \supp(\Phi_{2})}  (p\cdot \Phi_{1}(U) \cdot f(U)(x)) \big) \\
&\quad+
\big(\sum_{U \in \supp (\Phi_{2})\setminus \supp(\Phi_{1})}  ((1-p) \cdot \Phi_{2}(U) \cdot g(U)(x)) \big)\\
&\quad+
\Big(\sum_{U \in \supp (\Phi_{1})\cap \supp(\Phi_{2})}  
\big( ( p\cdot \Phi_{1}(U)  \cdot f(U)(x)) + ((1-p) \cdot \Phi_{2}(U) \cdot g(U)(x))\big)\Big) \\
&=\big(\sum_{U \in \supp (\Phi_{1})\setminus \supp(\Phi_{2})}  ((\Phi_{1}+_{p}\Phi_{2})(U) \cdot f(U)(x)) \big) \\
&\quad+
\big(\sum_{U \in \supp (\Phi_{2})\setminus \supp(\Phi_{1})}  ((\Phi_{1}+_{p}\Phi_{2})(U) \cdot g(U)(x)) \big)\\
&\quad+
\Big(\sum_{U \in \supp (\Phi_{1})\cap \supp(\Phi_{2})}  
\big( (\Phi_{1} +_{p} \Phi_{2})(U) \cdot (f(U)(x) +_{\frac{p\cdot \Phi_{1}(U)} {(\Phi_{1} +_{p} \Phi_{2})(U)}} g(U)(x))\big)\Big)
\tag{by $( p_{1}  \cdot q_{1}) + (p_{2} \cdot q_{2}) = (p_{1}  + p_{2}) \cdot (q_{1} +_{\frac{p_{1}} {p_{1}  + p_{2}}} q_{2})$, $\forall p_{1},p_{2},q_{1},q_{2}$}
\\
&= \sum_{U \in \supp (\Phi_{1}+_{p} \Phi_{2})} \big((\Phi_{1}+_{p} \Phi_{2})(U)\cdot h(U)(x)\big)
\end{align*}
Then, observe, that for every $U \in \supp (\Phi_{1}+_{p} \Phi_{2})$ we have $h(U) \in U$, since every $U$ is a convex set, and thus if $U$ contains $f(U)$ and $g(U)$ then it also contains $f(U) +_{q} g(U)$, for all $q$.
Thereby, we conclude 
\[d= \sum_{U \in \supp (\Phi_{1}+_{p} \Phi_{2})} \big((\Phi_{1}+_{p} \Phi_{2})(U)\cdot h(U)\big) \in \mu(S_{1}\mplus_{p} S_{2}).\]

We now prove the remaining inclusion, i.e., $\mu(S_1 \mplus_p S_2) \subseteq \mu(S_1) \mplus_p \mu(S_2)$

Let $\Phi\in S_1 \mplus_{p} S_2$ and let $d=\sum_{U \in \supp (\Phi)} \Phi(U)\cdot f(U)$, with $f:\supp(\Phi) \to \dset(X)$ such that $f(U) \in U$, be an element of $\mu(S_{1}\mplus_{p} S_{2})$.
Then, $\Phi=\Phi_{1} +_{p} \Phi_{2}$, with $\Phi_{1}\in S_{1},  \Phi_{2}\in  S_2$.
Then for every $x\in X$ we have
\begin{align*}
d (x)
&= \sum_{U \in \supp (\Phi_{1})\cup \supp(\Phi_{2})} ((\Phi_{1} +_{p} \Phi_{2})(U)\cdot f(U)(x))\\
&= \sum_{U \in \supp (\Phi_{1})\cup \supp(\Phi_{2})} ((p\cdot \Phi_{1}(U)\cdot f(U)(x)) + ((1-p)\cdot \Phi_{2}(U)\cdot f(U)(x)))\\
&= (\sum_{U \in \supp (\Phi_{1})} p\cdot \Phi_{1}(U)\cdot f(U)(x)) + (\sum_{U \in \supp (\Phi_{2})} (1-p)\cdot \Phi_{2}(U)\cdot f(U)(x))\\
&= (\sum_{U \in \supp (\Phi_{1})} \Phi_{1}(U)\cdot f(U)(x)) +_{p} (\sum_{U \in \supp (\Phi_{2})} \Phi_{2}(U)\cdot f(U)(x))
\end{align*}
which implies $d \in \mu(S_{1})\mplus_{p} \mu(S_{2})$.
\end{enumerate}

\end{proof}

In this way, we obtain a monad map $\iota\colon T_{\Sigma,E} \Rightarrow C$ defined as follows
$$\begin{array}{rcl}
\iota([x]_E)&=& \{\delta_x\}\\
\iota({[t_1\nplus t_2]_E}) & = &\iota([t_1]_E) \oplus \iota([t_2]_E)\\
\iota({[t_1+_pt_2]_E}) & = &\iota([t_1]_E)\mplus_p \iota([t_2]_E)
\end{array}$$
In the above definition, as well as in the remainder of the paper, we write $\iota$ in place of $\iota_X$ to simplify the notation.

Lemma \ref{lem:muhom}, together with the existence of unique bases, also allows us to derive a useful characterization of the multiplication $\mu$ of the monad $C$.

\begin{lemma}\label{lem:muub}
For $S\in CCX$,
$$\mu (S) = \conv\big(\bigcup_{\Phi \in UB(S)} \{ \sum_{U \in \supp \Phi} \Phi(U) \cdot d \mid d \in UB(U)\}\big).$$ 
\end{lemma}
\begin{proof}
We have $S=\conv (\bigcup_{\Phi \in UB(S)} \{\Phi\})$ which means that $S$ is a convex union of the sets $\{\Phi\}$, for $\Phi \in UB(S)$.
Then by Lemma \ref{lem:muhom} we derive
$\mu(S)=\conv (\bigcup_{\Phi \in UB(S)} \mu\{\Phi\})$. 
By definition, $\mu\{\Phi\} = \{ \sum_{U \in \supp (\Phi)} \Phi(U)\cdot d \mid d \in U\}$, hence
\begin{equation}\label{eq:muub1}
\mu(S)=\conv \big(\bigcup_{\Phi \in UB(S)} \{ \sum_{U \in \supp (\Phi)} \Phi(U)\cdot d \mid d \in U\}\big).
\end{equation}

Observe that the
Minkowski sum operation, which can be equivalently defined on arbitrary sets (i.e., not convex) of distributions, enjoys the following property:
\begin{equation}\label{eq:conv-mink}
\text{for any $S,T\subseteq X$, } \;\conv(S) \mplus_{p} \conv (T) = \conv(S \mplus_{p} T).
\end{equation}
Indeed, $S \mplus_{p} T\subseteq \conv(S) \mplus_{p} \conv (T)$, and as the Minkowski sum of convex sets is convex we have 
$\conv(S \mplus_{p} T)
\subseteq \conv(\conv(S) \mplus_{p} \conv (T))
=\conv(S) \mplus_{p} \conv (T).$
For the other direction, take 
$p(\sum_{i} p_{i} x_{i}) + (1-p)(\sum_{j} q_{j} y_{j}) \in \conv(S) \mplus_{p} \conv (T)$. We have:
\[p(\sum_{i} p_{i} x_{i}) + (1-p)(\sum_{j} q_{j} y_{j}) 
= p(\sum_{i,j} (p_{i} q_{j}) x_{i}) +(1-p) (\sum_{j} (p_{i} q_{j}) y_{j})
= \sum_{i,j} (p_{i} q_{j}) (p x_{i} + (1-p) y_{j})\]
which is then an element of  $\conv(S \mplus_{p} T)$. This proves (\ref{eq:conv-mink}).

For every $\Phi$, the set $\{ \sum_{U \in \supp (\Phi)} \Phi(U)\cdot d \mid d \in U\}$ is a Minkowski sum over the elements $U$ of $\supp(\Phi)$, which are themselves convex sets satisfying $U= \conv(UB(U))$.
Then by (\ref{eq:conv-mink}) we derive:
\begin{equation}\label{eq:muub2}
\{ \sum_{U \in \supp (\Phi)} \Phi(U)\cdot d \mid d \in U\}= \conv (\{ \sum_{U \in \supp (\Phi)} \Phi(U)\cdot d \mid d \in UB(U)\}).
\end{equation}
By (\ref{eq:muub1}) and (\ref{eq:muub2}) it holds: 
\[\mu(S)=\conv \big(\bigcup_{\Phi \in UB(S)}  \conv (\{ \sum_{U \in \supp (\Phi)} \Phi(U)\cdot d \mid d \in UB(U)\})\big).\]
As shown in the proof of \cite[Lemma 38]{BSV19arxiv}, we have:
\[
\text{for any $S,T\subseteq X$, } \;\conv(\conv(S) \cup T) = \conv(S \cup T).
\]
Hence, we derive:
\begin{align*}
\conv \big(\bigcup_{\Phi \in UB(S)}  \conv (\{ \sum_{U \in \supp (\Phi)} \Phi(U)\cdot d \mid d \in UB(U)\})\big)\\
\quad\quad= \conv \big(\bigcup_{\Phi \in UB(S)}  \{ \sum_{U \in \supp (\Phi)} \Phi(U)\cdot d \mid d \in UB(U)\}\big).
\end{align*}
%
%
%
%

\end{proof}

\section{Proving the isomorphism}\label{sec:inverse}
So far we have constructed a monad map $\iota\colon T_{\Sigma,E} \Rightarrow C$.
In this section, we prove that such map is an isomorphism by exploting Theorem \ref{thm:uniquebase}.

We start with a simple observation: for each set $X$, there is a trivial injection $i_X\colon T_{\Sigma_P}(X) \to T_{\Sigma}(X)$. A term  in $T_{\Sigma}$ is said to be a \emph{purely probabilistic term} (p-term, for short) iff it lays in the image of $i$. Since two $p$-terms are equal in $E$ iff they are also equal in $E_P$, then there is also an injection from $T_{\Sigma_P,E_P}(X)$ to $T_{\Sigma,E}(X)$. We overload the notation and denote it also with $i_X\colon T_{\Sigma_P,E_P}(X) \to T_{\Sigma,E}(X)$.

\begin{lemma}\label{prop:switch}
Let $\{-\}_X \colon  \mathcal{D}(X) \rightarrow C(X)$ be the function  mapping every distribution $d$ into the convex set $\{d\}$. The following diagram commutes.
$$\xymatrix{
T_{\Sigma_P,E_P }X \ar[r]^{i_X} \ar[d]_{\iota_X^P}  & T_{\Sigma,E}X \ar[d]^{\iota_X}\\
\mathcal{D}X \ar[r]_{\{-\}_X} & CX
}$$
\end{lemma}
\begin{proof}
We prove by induction that $\{ \iota_X^P([t]_{E_P}) \}_X = \iota_X(i_X(([t]_{E_P})))$  for all $t \in T_{\Sigma_P}$.
If $t=x\in X$, then $\{ \iota_X^P([x]_{E_P}) \}_X =\{ \delta_x \}= \iota_X([x]_E  )=\iota_X(i_X(([t]_{E_P})))$. If $t=t_1+_p t_2$, then 
\begin{align*}
\{ \iota_X^P([t_1+_p t_2]_{E_P}) \}_X & = \{p\cdot \iota_X^P([t_1]_{E_P})+ (1-p)\cdot \iota_X^P([t_1]_{E_P}) \}\\
& = \{i_X^P([t_1]_{E_P})\} \mplus_p \{i_X([t_2]_{E_{P}})\}\\
& = \iota_X(i_X([t_1]_{E_P})) \mplus_p \iota_X(i_X([t_2]_{E_P}))\\
& = \iota_X([t_1]_E) \mplus_p \iota_X([t_2]_{E})\\
& = \iota_X([t_1+_pt_2]_{E})\\
& = \iota_X(i_X(([t_1+_pt_2]_{E_P})))
\end{align*}
\end{proof}

Recall that the monap map $\iota^P \colon T_{\Sigma_P,E_P} \Rightarrow \Dis$ defined in \eqref{eq:iotap} is an isomorphism.
We call $\kappa^P \colon D \Rightarrow  T_{\Sigma_P,E_P}$ its inverse. By exploiting $\kappa^P$ and Theorem \ref{thm:uniquebase}, it is easy to define a function $\kappa_X \colon C(X) \to  T_{\Sigma, E}(X)$ as follows: for $S\in C(X)$ with base $\{d_1, \dots, d_n\}$
\begin{equation}\kappa_X(S)=[\, i(\kappa^P(d_1))\nplus \dots \nplus i(\kappa^P(d_n))\,]_E\text{.}\end{equation}

\begin{proposition}\label{prop:inverse1}
$\iota \after \kappa = \id_C$
\end{proposition}
\begin{proof}
Let $S\in C(X)$ be a convex set with base $\{d_1,\dots, d_n\}$. By definition $$\kappa(S)=[i(\kappa^P(d_1))\nplus \dots \nplus i(\kappa^P(d_n))]_E$$ and $$\iota(\kappa(S))=\iota([i(\kappa^P(d_1))]_E)\oplus \dots \oplus \iota([i(\kappa^P(d_n))]_E).$$ By Lemma \ref{prop:switch}, $\iota([\kappa(S)]_E)=\{d_1\} \oplus \dots \oplus \{d_n\}$ which is exactly $S$.
\end{proof}

We are now left to prove that $\kappa  \after \iota = \id_{T_{\Sigma,E}}$. This means that that any term $t$ is in the equivalence class of $\kappa  \after \iota ([t]_E)$, which by definition of $\kappa$ is $[i(\kappa^P(d_1))\nplus\dots \nplus i(\kappa^P(d_n))]_E$ where $\{d_1, \dots, d_n\}$ is the base for the space $\iota([t]_E)$.

\medskip

The first step consists in showing that every term is equivalent, modulo $E$, with a term of a certain shape: a term $t\in T_\Sigma(X)$ is said to be in \emph{nondeterministic-probablistic form}, \emph{n-p form} for short, if there exists $t_1, \dots, t_n\in T_{\Sigma_P}(X)$ such that $t=i(t_1) \nplus \dots \nplus i(t_n)$. This can be thought of as an analogous of the disjunctive-conjunctive form that is commonly used in propositional logic.

\begin{example}
The term $(x\nplus y)+_{\frac{1}{2}}(y+_{\frac{1}{3}} z)$ is not in n-p form, since $x\nplus y$ occurs inside $+_{\frac{1}{2}}$. However, by using the distributivity axiom $(D)$, we have that $(x\nplus y)+_{\frac{1}{2}}(y+_{\frac{1}{3}} z) =_E (x+_{\frac{1}{2}} (y+_{\frac{1}{3}} z)) \nplus (y+_{\frac{1}{2}} (y+_{\frac{1}{3}} z))$ which is in n-p form.
\end{example}
The following proposition ensures that every term is equivalent through $E$ to one in n-p form.

\begin{proposition}\label{prop:npform} 
For all $t\in T_{\Sigma}(X)$, there exists $t'$ in n-p form such that $t=_Et'$.
\end{proposition}
\begin{proof}
Intuitively, by virtue of the axiom $(D)$ all the occurrences of $+_p$ can be pushed inside some $\nplus$.  This can be proved formally by means of the following term rewriting system.
$$ (t_1\nplus t_2) +_p t_3 \rightsquigarrow (t_1+_p t_3) \nplus  (t_2 +_p t_3) \qquad t_1+_p (t_2 \nplus t_3) \rightsquigarrow (t_1+_p t_2) \nplus (t_1 +_p t_3) $$
If $t\in T_\Sigma(X)$ rewrites to $t'\in T_\Sigma(X)$, then $t=_E t'$ since the left rule is just the axiom (D), while the right can be derived using $(C_p)$, $(D)$ and $(C_p)$ again. 

Using standard term rewriting techniques from \cite{dershowitz1982orderings} we can prove that the rewriting system terminates: 
\begin{itemize}
	\item[(1)] Define the partial order $+_p > +$ on $\Sigma$; 
	\item[(2)] Observe that the generated recursive path ordering on $T_{\Sigma}(X)$ is a simplification ordering (see e.g., Example A in Section 5 of \cite{dershowitz1982orderings});
	\item[(3)] Conclude by the First Termination Theorem.
\end{itemize}

Finally, we observe that a term $t$ is in n-p form iff  $t\not \rightsquigarrow$: Indeed, if $t$ is in n-p form then there is no redex for the two rules above. On the other hand, if $t$ is not in n-p form, then some $+_p$ should occur inside a $\nplus$ and then one of the rules applies.

Therefore, each term $t$ can be rewritten into an $E$-equivalent term $t'$ in n-p form.
\end{proof}

Given  a term $t'\in T_\Sigma (X)$ in n-p form and $t_1, \dots, t_n\in T_{\Sigma_P}(X)$ such that $t'=i(t_1) \nplus \dots \nplus i(t_n)$, one would like $\{ \iota^P([t_1]_E), \dots, \iota^P([t_n)]_E\}$ to be the base for $\iota([t']_E)$. But this is not always the case since some $\iota^P([t_i]_E)$ can be in the convex combination of the other $\iota^P([t_j]_E)$.


\begin{example}
The term $(x+_{\frac{1}{2}}y)\nplus (x+_{\frac{2}{3}}(x\nplus y))$ is not in n-p form. By applying the rewriting procedure in the proof of Proposition \ref{prop:npform} one obtains: $(x+_{\frac{1}{2}}y)\nplus (x+_{\frac{2}{3}}(x+y))=_E (x+_{\frac{1}{2}}y)\nplus (x+_{\frac{2}{3}}x) \nplus (x+_{\frac{2}{3}}y)$. Observe that this is equivalent to $(x+_{\frac{1}{2}}y)\nplus x \nplus (x+_{\frac{2}{3}}y)$. The convex set $\iota((x+_{\frac{1}{2}}y)\nplus x \nplus (x+_{\frac{2}{3}}y))$ has base $\{\iota^P([x+_{\frac{1}{2}}y]_P, \iota_P([x]_P))\}=\{ \frac{1}{2}x + \frac{1}{2}y, \delta_x \}$. Indeed  the distribution $\iota^P(x+_{\frac{2}{3}}y)= \frac{2}{3}x + \frac{1}{3}y $ is a convex combination of $\{ \frac{1}{2}x + \frac{1}{2}y, \delta_x \}$ as $\frac{2}{3}x + \frac{1}{3}y = \frac{2}{3}(\frac{1}{2}x + \frac{1}{2}y) + \frac{1}{3} \delta_x$. 
\end{example}

The next three lemmas are necessary to show that, using the axioms in $E$, we can remove from $t'$ those summands $i(t_i)$ such that 
$\iota^P([t_i]_E)$ is in the convex combination of the other $\iota^P([t_j]_E)$. These are again partly from~\cite{BSV19arxiv}.

\begin{lemma}[Convexity law]\label{lemma:convexitylaw}
For all terms $t_1,t_2\in T_{\Sigma}(X)$, for all $p\in(0,1)$, $$t_1 \nplus t_2 =_E t_1 \nplus t_2 \nplus  (t_1 +_p t_1)\text{.}$$
\end{lemma}
\begin{proof}
First, we observe that 
\begin{equation}\label{eq:theory}
t_1\nplus t_2=t_1\nplus (t_2+_p t_1) \nplus  (t_1+_p t_2) \nplus  t_2 
\end{equation}
as proved by the following derivation.
$$\begin{array}{rcl}
t_1\nplus t_2 & \stackrel{(I_p)}{=}& (t_1\nplus t_2)+_p(t_1\nplus t_2) \\
 &  \stackrel{(D)}{=} & ( (t_1\nplus t_2)+_p t_1  ) \nplus  ( (t_1\nplus t_2) +_p t_2) \\
 & \stackrel{(D)}{=} &((t_1+_pt_1) \nplus  (t_2+_p t_1)) \nplus  ( (t_1+_pt_2) \nplus  (t_2+_pt_2)  )\\
 &  \stackrel{(I_p)}{=} & t_1\nplus (t_2+_p t_1) \nplus  (t_1+_p t_2) \nplus  t_2 
 \end{array}$$
Then we conclude with
$$\begin{array}{rcl}
t_1 \nplus  t_2 \nplus  (t_1 +_p t_2) & \stackrel{\eqref{eq:theory}}{=}  & t_1\nplus (t_2+_p t_1) \nplus  (t_1+_p t_2) \nplus  t_2 \nplus  (t_1+_p t_2)\\& \stackrel{(I_p)}{=} & t_1\nplus (t_2+_p t_1) \nplus  (t_1+_p t_2) \nplus  t_2
 \end{array}$$
\end{proof}

\begin{lemma}\label{lemma:convps}
Let $t,t_1, \dots, t_n \in T_{\Sigma_P}(X)$ such that $\iota^P([t]_P) \in conv\{\iota^P([t_1]_P), \dots, \iota^P([t_n]_P)\}$. Then there exist $p_1\dots p_{n-1}\in (0,1)$ such that $t =_P (\dots (t_1+_{p_1}t_2)+_{p_2} \dots )+_{p_{n-1}} t_n $.
\end{lemma}
\begin{proof}
If $\iota^P([t]_P) \in \conv\{\iota^P([t_1]_P), \dots, \iota([t_n]_P)\}$, then $\iota^P([t]_P)= \mu^\Dis (\sum_i q_i\cdot \iota^P([t_i]_P))$. Since $\iota^P$ is a monad map, its inverse $\kappa^D\colon \Dis \Rightarrow T_{\Sigma_P,E_P}$ is also a monad map and in particular, it makes the following diagram commutes.
$$\xymatrix{
\Dis\Dis X \ar[d]_{\mu_X^\Dis} \ar[r]^{\Dis \kappa^P_X}& \Dis T_{\Sigma_P,E_P}X \ar[r]^{\kappa^P_{T_{\Sigma_P,E_P}}}& T_{\Sigma_P,E_P} T_{\Sigma_P,E_P}X \ar[d]^{\mu_X^{T_{\Sigma_P,E_P}}} \\
\Dis X \ar[rr]_{\kappa_X^P}&& T_{\Sigma_P,E_P}
}$$
Therefore, we have that
\begin{align*}
[t]_P& = \kappa^P \circ \iota^P([t]_P)\\
& = \kappa^P \circ ( \mu^\Dis (\sum_i q_i\cdot \iota^P([t_i]_P)))\\
& = \mu^{T_{\Sigma_P,E_P}} \circ \kappa^P \circ \Dis \kappa  (\sum_i q_i\cdot \iota^P([t_i]_P))\\
& = \mu^{T_{\Sigma_P,E_P}} \circ \kappa^P  (\sum_i q_i\cdot \kappa^P\circ \iota^P([t_i]_P))\\
& = \mu^{T_{\Sigma_P,E_P}} \circ \kappa^P  (\sum_i q_i\cdot [t_i]_P)
\end{align*}
Observe that $\sum_i q_i\cdot [t_i]_P \in \Dis T_{\Sigma_p,E_P} (X)$ and that $\kappa^P_{T_{\Sigma_p,E_P} X}$ maps it into an element of  $T_{\Sigma_p,E_P}  T_{\Sigma_p,E_P} (X)$, namely a term obtained by the operations $+_p$ and the constants $[t_i]_P$. Thanks to the axioms in $E_P$ any such term can always be written as  $(\dots ([t_1]_P+_{p_1}[t_2]_P)+_{p_2} \dots )+_{p_{n-1}} [t_n]_P $ for some $p_i\in (0,1)$. Then, the application of $\mu^{T_{\Sigma_P,E_P}}$ to $[(\dots ([t_1]_P+_{p_1}[t_2]_P)+_{p_2} \dots )+_{p_{n-1}} [t_n]_P ]_P$ gives just $[(\dots (t_1+_{p_1}t_2)+_{p_2} \dots )+_{p_{n-1}} t_n]_P$. Thus $t =_P (\dots (t_1+_{p_1}t_2)+_{p_2} \dots )+_{p_{n-1}} t_n $.
\end{proof}

\begin{lemma}\label{lemma:removingconv}
Let $t,t_1, \dots, t_n \in T_{\Sigma_P}(X)$ such that $\iota^P([t]_P) \in \conv\{\iota^P([t_1]_P), \dots, \iota^P([t_n]_P)\}$. Then 
$$i(t_1) \nplus   \dots \nplus  i(t_n) =_E i(t_1) \nplus   \dots \nplus  i(t_n) \nplus  i(t)$$
\end{lemma}
\begin{proof}
By Lemma \ref{lemma:convps}, we take $p_1,\dots , p_{n-1}$ such that
\begin{equation}\label{eq:d}
t =_P (\dots (t_1+_{p_1}t_2)+_{p_2} \dots )+_{p_{n-1}} t_n\text{.}
\end{equation}
By Lemma \ref{lemma:convexitylaw}, $i(t_1) \nplus   \dots \nplus  i(t_n)$ is $E$-equivalent to 
$i(t_1) \nplus   \dots \nplus  i(t_n) \nplus  i(t_1+_{p_1}t_2)$. By applying Lemma \ref{lemma:convexitylaw} again, one obtains $i(t_1) \nplus   \dots \nplus  i(t_n) \nplus  i(t_1+_{p_1}t_2) \nplus  i((t_1+_{p_1}t_2)+_{p_2}t_3) $. We can then remove $i(t_1+_{p_1}t_2)$ using Lemma \ref{lemma:convexitylaw}, to obtain 
 $$i(t_1) \nplus   \dots \nplus  i(t_n)  \nplus i((t_1+_{p_1}t_2)+_{p_2}t_3)  \text{.}$$
 By iterating this procedure, one obtains
 $$i(t_1) \nplus   \dots \nplus  i(t_n) \nplus  i((\dots (t_1+_{p_1}t_2)+_{p_2} \dots )+_{p_{n-1}} t_n)$$
 which, by \eqref{eq:d}, is  $ i(t_1) \nplus   \dots \nplus  i(t_n) \nplus  i(t)$.
\end{proof}

\begin{proposition}\label{prop:normalform}
For all terms $t\in T_{\Sigma} (X)$, there exist $t_1, \dots, t_n\in T_{\Sigma_P}$ such that
$$t=_E i(t_1)\nplus \dots \nplus i(t_n)$$
and $\{\iota^P([t_1]_P), \dots, \iota^P([t_n]_P) \}$ is the base of $\iota([t]_E)$.
\end{proposition}
\begin{proof}
By Proposition \ref{prop:npform}, there exists $t' \in T_\Sigma(X)$ in n-p form such that $t=_E t'$.
Take $t_1', \dots, t_m'\in T_{\Sigma_P}$ such that $$t'=i(t_1)\nplus \dots \nplus i(t_m)\text{.}$$
By definition of $\iota$, $\iota([t]_E)=\iota(i([t_1]_P)) \oplus \dots \oplus \iota(i([t_m]_P))$ which by Lemma \ref{prop:switch} is $\{ \iota^P([t_1]_P) \} \oplus \dots \oplus \{ \iota^P([t_m]_P) \}$. By definition of $\oplus$, this is just $\conv \{\iota^P([t_1]_P), \dots, \iota^P([t_m]_P) \}$. Therefore, to conclude that $\{\iota^P([t_1]_P), \dots, \iota^P([t_m]_P) \}$ is the base of $\iota([t]_E)$ we only need to show that none of the $\iota^P([t_i]_P)$ is in the convex combination of the others $\iota^P([t_j]_P)$. This is not true in general, but thanks to Lemma \ref{lemma:removingconv} all such $t_i$ can be removed, while preserving $E$-equivalence. To be more precise, by associativity and commutativity of $\nplus $, we can assume that $\iota^P([t_1]_P), \dots, \iota^P([t_n]_P)$ form the base, while $\iota^P([t_{n+1}]_P), \dots, \iota^P([t_m]_P)$ are in $\conv \{\iota^P([t_1]_P), \dots, \iota^P([t_n']_P) \}$. Then, by repeating $(m-n)$-times  Lemma \ref{lemma:removingconv}, we conclude that 
$t'=_E i(t_1)\nplus \dots \nplus i(t_{n})$.
\end{proof}

\begin{proposition}\label{prop:inverse2}
$\kappa \after \iota   = \id_{T_{\Sigma,E}}$
\end{proposition}
\begin{proof}
We need to prove that for all terms $t\in T_\Sigma(X)$, $[t]_E= \kappa \after \iota ([t]_E)$. 
By Proposition \ref{prop:normalform}, there exists $t_1, \dots, t_n \in T_{\Sigma_P}(X)$  such that 
$$t=_Ei(t_1)\nplus \dots \nplus i(t_n)$$

and  $\{\iota^P([t_1]_P), \dots, \iota^P([t_n]_P)\}$ is the base for $\iota([t]_E)$.

By definition of $\kappa$,  $\kappa(\iota([t]_E))$ is exactly $[i(\kappa^P \circ \iota^P[t_1]_P)\nplus \dots \nplus i(\kappa^P \circ \iota^P[t_n]_P)]_E = [t]_E$.
\end{proof}

This is enough to conclude the proof of Theorem \ref{thm:presentation}. Indeed we have that $\iota\colon T_{\Sigma,E} \Rightarrow C$ is a monad map and that, by Propositions \ref{prop:inverse1} and \ref{prop:inverse2}, it is an isomorphism.

\bibliography{biblio}

\end{document}